\documentclass[letterpaper, 10 pt, conference]{ieeeconf}
\IEEEoverridecommandlockouts 

\usepackage[utf8]{inputenc} 
\usepackage[T1]{fontenc}    
\usepackage{url}            
\usepackage{amsfonts}       
\usepackage{amsmath} 
\usepackage{amssymb}  
\usepackage{mathrsfs}
\usepackage{bm}
\usepackage{graphicx}
\usepackage{nicefrac}       

\newcommand{\mC}{\mathcal{C}}

\newcommand{\mL}{\mathcal{L}}

\newcommand{\mU}{\mathcal{U}}

\newcommand{\mX}{\mathcal{X}}

\newcommand{\R}{\mathbb{R}}
\newcommand{\mbE}{\mathbb{E}}
\newcommand{\mbP}{\mathbb{P}}

\newcommand{\mO}{\mathcal{O}}
\newcommand{\tmO}{\tilde{\mathcal{O}}}
\newcommand{\tx}{\tilde{x}}

\newcommand{\innerp}[1]{\langle #1 \rangle}

\newcommand{\prob}[1]{\mathbb{P}\left( #1 \right)}
\renewcommand{\top}{\mathsf{T}}

\usepackage[dvipsnames]{xcolor}

\newtheorem{definition}{Definition}[section]

\newtheorem{thm}{Theorem}

\newtheorem{assumption}{Assumption}
\newtheorem{example}{Example}
\newtheorem{remark}{Remark}[section]

\title{\LARGE \bf
Model Predictive Control with High-Probability Safety Guarantee for Nonlinear Stochastic Systems 
}
\author{Zishun Liu$^{1}$, Liqian Ma$^{1}$ and Yongxin Chen$^{1}$
\thanks{$^{1}$The authors are with Georgia Institute of Technology, Atlanta, GA 30332. 
        {\tt\small \{zliu910\}\{mlq\}\{yongchen\}@gatech.edu}}%
}

\begin{document}
\maketitle
\thispagestyle{empty}
\pagestyle{empty}

\begin{abstract}
We present a model predictive control (MPC) scheme for nonlinear stochastic systems that renders a safety guarantee with high probability. Unlike most existing stochastic MPC schemes, our method adopts a set erosion strategy that converts the probabilistic safety constraint into a tractable deterministic safety constraint on a smaller safe set under deterministic dynamics. The latter can be addressed using any scalable, off-the-shelf, deterministic MPC algorithm. 
The key to the effectiveness of the proposed MPC is a tight bound on the stochastic fluctuation of a stochastic trajectory around its nominal version. With this tight bound, our method can guarantee a safety bound with a high probability level (e.g., 99.99\%), making it particularly suitable for safety-critical applications involving complex nonlinear dynamics. 
Finally, we present a rigorous analysis of the closed-loop dynamics, which resolves a subtle theoretical gap in the existing literature.
Numerical experiments are provided to validate the effectiveness of the proposed MPC method.
\end{abstract}

\section{Introduction}
\label{sec:introduction}
 
Safety is a fundamental requirement in real-world systems. As many real-world systems are not naturally safe, an important task is to synthesize a controller so that the controlled dynamics satisfy the safety requirements \cite{lavaei2022automated}. For deterministic systems with bounded uncertainties, safe controller synthesis seeks to keep the closed-loop system invariant in the safe region in the worst-case scenarios \cite{kouvaritakis2015model}. This task has been addressed by a variety of methods, including control barrier functions (CBF) \cite{ames2019control}, reachability analysis \cite{bajcsy2019efficient}, and model predictive control (MPC).  Among all these methods, MPC has been flourishing for decades due to its ability to effectively cope with complex system dynamics and constraints \cite{morari1999model}. For deterministic systems, MPC ensures safety by formulating the worst-case safety condition into robustness constraints, and has shown exceptional success in various applications like autonomy \cite{lew2024risk,knoedler2025safety}.

However, many real-world systems are also prone to stochastic disturbances \cite{santoyo2021barrier}, which are usually unbounded (e.g., Gaussian noise) or rarely reach the worst case. Consequently, the objective of safe controller synthesis becomes to ensure the entire stochastic trajectory satisfies safety constraints \textit{with high probability}. Existing efforts have been made to adapt the aforementioned deterministic methods to this problem, including stochastic CBF \cite{fushimi2025safety}, stochastic reachable tube \cite{APV-MKO:21}, and stochastic MPC. In this paper, we focus on stochastic MPC \cite{liniger2017racing,yin2025safe}.

In general, stochastic MPC formulates safety conditions as chance constraints that demand satisfaction with some given probability level.
Methods such as sample-based MPC \cite{6426462,wu2022robust} and time-varying linearized MPC \cite{yu2023stochastic} have been applied in practice. Besides, given fruitful results in deterministic MPC, a natural attempt to handle the chance constraints is to convert them into deterministic ones. 
When the system is linear with Gaussian or sub-Gaussian noise, the chance constraints can be converted to deterministic constraints on the mean and covariance of the trajectory \cite{farina2016stochastic,ao2025stochastic}. However, this strategy is restricted to linear systems. For nonlinear dynamics, an effective strategy is called set erosion \cite{liu2024safety}, which converts the chance constraints to deterministic safety constraints on the nominal trajectory (the trajectory with no disturbance) in an eroded safe region. Existing works such as \cite{hewing2020recursively,kohler2024predictive} have proposed tractable stochastic MPC schemes based on this strategy, but they have some limitations. First, many of them only provide probabilistic safety insurance at each single time instant, but not \textit{for the entire trajectory}. Second, the erosion depth of the safe regions derived in existing MPC schemes is overly conservative when the safety level is high (e.g., safe with probability $>99.9\%$).

In this paper, we propose an MPC framework for stochastic nonlinear systems with a high-probability safety guarantee. The proposed MPC framework is based on the tractable MPC in \cite{kohler2024predictive}, which
leverages the set erosion strategy that converts the chance constraints to deterministic safety constraints. Compared to \cite{kohler2024predictive}, we significantly improve the feasibility of the proposed MPC scheme when the probability level of safety is high (e.g., 99.99\%). This is achieved by adopting a sufficiently tight bound on the erosion depth of the safe set, which relaxes the converted deterministic safe constraints while theoretically guaranteeing the trajectory-level safety. Moreover, we justify that the bound on the set erosion depth depends on the Lipschitz constant of the open-loop system, and the change of the closed-loop Lipschitz constant due to the implicit feedback of the MPC is irrelevant. This property is crucial since it is difficult to track the changing closed-loop Lipschitz constant in practice.
Note that in previous work along this line \cite{kohler2024predictive}, whether the suitable set erosion depth depends on the closed-loop Lipschitz constant or the open-loop one is unclear.

\textit{Notations.} The set of positive integers is denoted by $\mathbb{N}_{+}$. We use $\|\cdot\|$ to denote $\ell_2$ norm. Given two sets $A,B\subseteq \R^n$, $A\oplus B$ denotes their Minkowski sum, and their Minkowski difference is denoted by $A\ominus B$. We use $\mbE$ to denote expectation, $\mbP$ to denote probability, $\mathcal{N}(\mu,\Sigma)$ to denote Gaussian distribution, and $\mathcal{B}^n(r,y)$ to denote the ball $\{x\in\R^n: \|x-y\|\leq r\}$. We say a random variable $X\sim \mathcal{G}$ if $X$ is i.i.d. sampled from the distribution $\mathcal{G}$. For a system state $x_t$, its trajectory is denoted by $\{x_t\}$.

\section{Problem Statement} \label{sec: problem}
To begin with, we introduce the safe control problem of stochastic systems and review the standard stochastic MPC method for it.

\subsection{Safety of Stochastic Systems}
Consider the discrete-time stochastic system
\begin{equation}\label{sys: d-t ss}
     X_{t+1}=f(X_t,u_t)+w_t,
\end{equation}
where $X_t\in \R^n$ is the system state, $u_t\in\mU\subset\R^p$ is an open-loop control input chosen from an input set $\mU$, $w_t\in\R^n$ is a stochastic disturbance, and $f: \R^n\times\R^p\times\mathbb{N}_+\to\R^n$ is a nonlinear transition function. In this paper, we impose the Lipschitz nonlinearity condition on the open-loop system. 

\begin{assumption}\label{ass: Lipschitz f}
    There exists $L\geq0$ such that $\|f(x,u) - f(y,u)\| \leq L\|x-y\|$ holds for every $x,y$ in the domain of $f$ and every $u\in\mU$.
\end{assumption}
\begin{remark} \label{remark}
    For nonlinear systems, it is sometimes difficult to capture their global Lipschitz constant, or the global Lipschitz constant can be overly large. In this case, one can set $u_t=K(X_t)+\tilde{u}_t$ and bound the Lipschitz constant of $\tilde{f}(X_t,\tilde{u}_t)=f(X_t,K(X_t)+\tilde{u}_t)$ in certain subsets of the state space. Some methods of designing the feedback $K(X_t)$ include control contraction metrics \cite{manchester2017control}, learning-based methods \cite{9683354}, etc.
\end{remark}

We model $w_t$ as a \textit{sub-Gaussian} disturbance, which covers a wide range of noise distributions such as Gaussian, uniform, and any zero-mean noise with bounded support \cite{vershynin2018high}.
\begin{assumption}\label{ass: bounded sigma}
   For the system~\eqref{sys: d-t ss}, $w_t\sim subG(\sigma^2)$ with some finite $\sigma>0$ for $\forall t\geq0$. That is, $\mbE(w_t)=0$ and
$\mbE_{w_t}\left(e^{\lambda \innerp{l,w_t}}\right)\leq e^{\frac{\lambda^2\sigma^2}{2}}$ holds for all $\lambda\in\R$ and $l$ on the $\ell_2$-unit sphere.
\end{assumption}

Our goal is to design a controller for the stochastic system \eqref{sys: d-t ss} with safety guarantee. Unlike deterministic safety in the sense of forward invariance \cite{ames2019control}, the safety of stochastic systems is formalized as the following chance constraint to better capture the effect of stochastic noise.
\begin{definition} \label{def: sto safety}
    Consider stochastic system \eqref{sys: d-t ss} with input set $\mathcal{U}\subseteq\R^p$. Given a time horizon $T$, a control strategy $u_t\in\mU$, a probability level $\delta\in(0,1)$, a safe set $\mC\subseteq\mathbb{R}^n$, and an initial configuration $\mathcal{X}_0\subseteq\R^n$, the system is said to be \textit{safe with $1-\delta$ guarantee} during $t\leq T$ if $\mX_0\subseteq\mC$ and $ X_0\in\mathcal{X}_0~ \Rightarrow~\prob{X_t\in\mC,~ \forall t\leq T}\geq 1-\delta.$
\end{definition}
\smallskip

With the definition of stochastic safety, our objective can be stated as developing a controller that provides $1-\delta$ safety guarantee for the stochastic system \eqref{sys: d-t ss}. Given the safety probability level $1-\delta$, terminal time $T$, cost functions $\mathcal{L}_t(\cdot,\cdot)$, terminal cost $\Phi(\cdot)$, safe set $\mC$ and initial state $X_0\in\mX_0\subseteq\mC$, the control task is formulated as follows
\begin{subequations} \label{origin opt ctrl}
      \begin{align}
        \min_{u_0,\dots, u_{T-1}} J&=  \mbE\left(\sum_{t=0}^{T-1}\mL_t(X_t,u_t)+\Phi(X_T)\right)  \\
        \mbox{s.t.}\quad & X_{t+1}=f(X_{t},u_t)+w_t \\
        & \prob{X_t\in\mC,~\forall t\leq T}\geq 1-\delta \\
        & u_{t}\in\mU,~\forall t=0,\dots,T-1.
    \end{align}
 \end{subequations}

 Unfortunately, it can be intractable to directly solve Problem \eqref{origin opt ctrl}. It requires evaluating the stochastic safety constraint (\ref{origin opt ctrl}c) over the entire time horizon, which is typically implemented by sampling-based methods and thus is computationally heavy \cite{kohler2024predictive}. To overcome this challenge, we need to convert \eqref{origin opt ctrl} to a tractable control scheme.

\subsection{Stochastic MPC with Chance Constraint}
 
 In this paper, we focus on handling \eqref{origin opt ctrl} through stochastic MPC, as it systematically seeks control objectives and probabilistic constraint satisfaction \cite{7740982}. At each time $t$, given a prediction window $m$, a typical stochastic MPC predicts an $m$-length trajectory $\{z_{t+k|t}\}$, $k=0,\dots,m$ by solving the stochastic optimization problem with system dynamic constraints $z_{t+k+1|t}=f(z_{t+k|t},u_{t+k|t})$$ + w_{t+k|t}$, where $w_{t+k|t}\sim subG(\sigma^2)$, and chance constraints for safety guarantee. 
However, it is challenging to construct a proper chance constraint for the predicted trajectory. One such construction in \cite{hewing2020recursively,kohler2024predictive} is $\prob{\hat{x}_{k|t}\in\mC}\geq 1-\delta$, $\forall k\leq T$, where $\hat{x}_{k+1|t}=f(\hat{x}_{k|t},u_{t+k|t})+w_{t+k|t}$ with $\hat{x}_0=x_0$. However, this constraint is imposed on each single time point of the predicted trajectory, rather than the \textit{entire actual trajectory}, so it fails to provide a $1-\delta$ safety guarantee for the system. Moreover, chance constraints for nonlinear systems cannot be directly handled by typical statistical methods except Monte-Carlo sampling, which is computationally intractable.

\section{Stochastic MPC with High-Probability Safety Guarantee} \label{sec: mpc}
Departing from most existing stochastic MPC methods, we adopt a strategy that converts the chance constraint into a deterministic safety constraint and reduce the stochastic MPC problem into a deterministic one. In particular, following \cite{kohler2024predictive,liu2024safety}, we leverage a technique known as set erosion that ensures $1-\delta$ probabilistic safety guarantee by ensuring a deterministic safety guarantee on an eroded safe set.

\subsection{Tractable Set Erosion-Based MPC} \label{sec: erosion}
For a stochastic trajectory of \eqref{sys: d-t ss} with initial state $X_0$ and input $u_t$, its associated nominal trajectory $\{x_t\}$ is given as
\begin{equation}\label{eq: nominal}
    x_{t+1}=f(x_t,u_t),~ x_0=X_0.
\end{equation}
Intuitively, a stochastic trajectory fluctuates around its associated nominal trajectory with high probability. Therefore, if we erode the safe set $\mC$ with a suitable depth $r_{\delta,t}$ to obtain an eroded subset $\tilde{\mC}_t=\mC\ominus\mathcal{B}^n(r_{\delta,t},0)$ and control the nominal trajectory $\{x_t\}$ to stay within $\tilde{\mC}_t$, then the associated stochastic trajectory $\{X_t\}$ is expected to stay safe with high probability. This strategy is termed \textit{set erosion} \cite{liu2024safety}. 

Finding an appropriate erosion depth $r_{\delta,t}$ is the key of the set erosion strategy. A natural choice of $r_{\delta,t}$ is the \textit{radius of probabilistic tube (PT)}. Given a finite time horizon $[0,\,T]$ and a probability level $\delta\in(0,1)$, a curve $r_{\delta,t}:(0,1)\times[0,T]\to\R_{\geq0}$ is said to be the \textit{radius of PT}, if for any associated trajectories $\{X_t\}$ and $\{x_t\}$:
\begin{equation} \label{eq: def PT}
    \prob{\|X_t-x_t\|\leq r_{\delta,t},~\forall t\leq T}\geq 1-\delta.
\end{equation}     

It should be noted that $r_{\delta,t}$ is a bound applied to the \textit{entire stochastic trajectory}, which means the probability that \textit{none of} the state on the trajectory violates the bound is at least $1-\delta$. It can be shown that when the erosion depth is chosen as the radius of PT, the safety of $\{x_t\}$ on $\mC\ominus\mathcal{B}^n(r_{\delta,t},0)$ yields the satisfaction of the chance constraint in Definition \ref{def: sto safety} for the associated $\{X_t\}$ \cite{liu2025safety}. 

The set erosion strategy has been considered in several existing works \cite{kohler2024predictive,hewing2020recursively,liu2025safety}, among which our MPC framework is based on the tractable stochastic MPC proposed in \cite[Section IV]{kohler2024predictive}. Given the prediction window $m$, cost functions $\mL_t$, terminal cost $\Phi$ and a terminal state set $\mX_f$, the proposed MPC solves the following optimization problem at each time $t$ with variable $\bm{u}_t=\{u_{t+k|t}\}$, $k=0,\dots,m-1$. 
\begin{subequations}\label{eq: mpc scheme}
    \begin{align}
        &\min_{\bm{u}_t} J_t= \sum_{k=0}^{m-1} \mL_{t+k}(z_{t+k|t},u_{t+k|t})+\Phi(z_{t+m|t}) \\
        \mbox{s.t.}\quad & z_{t|t}=X_t \\
        &z_{t+k+1|t}=f(z_{t+k|t},u_{t+k|t}) \\
        & \tx_{t|t}=x_t \\
        & \tx_{t+k+1|t}=f(\tx_{t+k|t},u_{t+k|t}) \\
        & \tx_{t+k|t}\in\mC\ominus\mathcal{B}^n\left(r_{\delta,t+k},0\right) \\
        & \tx_{t+m|t}\in \mX_f \\
        & u_{t+k|t}\in\mU,~ k=0,\dots,m-1,
    \end{align}
\end{subequations}
 where $\{z_{t+k|t}\}$ is the prediction of the actual trajectory, $\{\tx_{t+k|t}\}$ is the prediction of the nominal trajectory, $J_t$ is the cost function of $\{z_{t+k|t}\}$ and $\bm{u}_t$, and $r_{\delta,t}$ is the radius of PT \eqref{eq: def PT}. $\{z_{t+k|t}\}$ is calculated in a noise-free environment so that the optimization problem is deterministic and can be solved efficiently.  
 The terminal state constraint (\ref{eq: mpc scheme}f) is added to ensure the recursive feasibility of \eqref{eq: mpc scheme}.
 At each time $t$, the MPC scheme receives the current state $X_t$, the nominal state $x_t$, and the offline computed bound $r_{\delta,\tau}$, $\tau=0,\dots,T$. Then the problem \eqref{eq: mpc scheme} is initialized by these data and solved. Suppose that a solution $\bm{u}_t^*$ is found, then its first element $u_{t|t}^*$ is chosen as the control input $u_t$ to $X_t$ and $x_t$ simultaneously. When $t$ is close to $T$, the prediction window is reduced to ensure that $t+m\leq T$. An illustration of the presented MPC is shown in Figure \ref{fig: mpc}. We point out that $u_t$ is a closed-loop policy since the state $X_t$ is used in determining the control.

\begin{figure}[t] 
	\centering
  \includegraphics[width =0.9\linewidth]{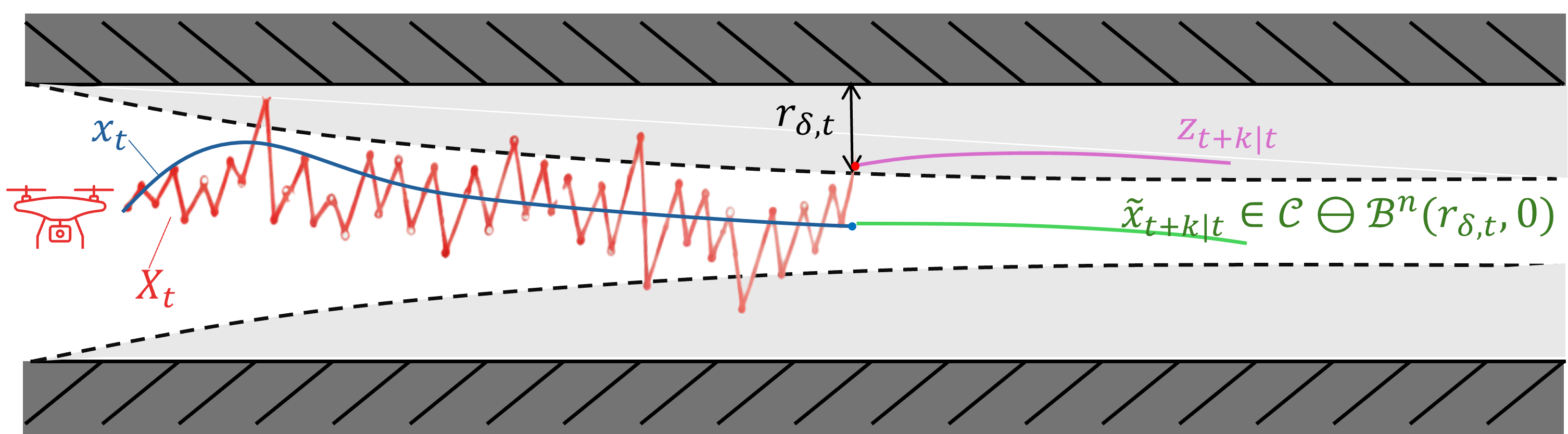}
	\caption{An illustration of the presented stochastic MPC scheme. Blue curve: current nominal trajectory $\{x_t\}$. Green curve: prediction of $\{x_t\}$. Red curve: Current actual trajectory $\{X_t\}$. Purple curve: Prediction of $\{X_t\}$. The gray areas are the eroded part of the safe set $\mC$. 
 }
 \label{fig: mpc}
 \end{figure}
 
By utilizing the set erosion strategy, the chance constraint is converted to (\ref{eq: mpc scheme}f), which is a deterministic constraint on the predicted nominal trajectory $\{\tx_{t+k|t}\}$ given the current information. Therefore, \eqref{eq: mpc scheme} is a tractable deterministic optimization problem as long as the Minkowski difference in (\ref{eq: mpc scheme}f) can be handled efficiently. In practice, $\mC$ is the complementary of the unsafe region $\mC_u$, so (\ref{eq: mpc scheme}f) is equivalent to $\tx_{t+k|t}\notin \mC_u\oplus\mathcal{B}^{n}(r_{\delta,t},0)$. When $\mC_u$ is the union of obstacles in the shapes of ellipsoids and polyhedra, $\mC_u\oplus\mathcal{B}^{n}(r_{\delta,t},0)$ can be efficiently computed, corresponding to similar obstacles of larger size. Thus, \eqref{eq: mpc scheme} can be solved with well-studied obstacle-avoidance control methods for deterministic systems.

Given the MPC scheme, it is important to understand its closed-loop properties. To ensure the feasibility of the optimization problem \eqref{eq: mpc scheme}, we impose the following assumption on the terminal set $\mX_f$ and the control input set $\mU$. 
\begin{assumption} \label{ass: Xf}
 \text{ }
\begin{enumerate}
    \item $\mX_f\in\mC\ominus\mathcal{B}^n\left(r_{\delta,t},0\right)$, $\forall t=0,\dots,T$, and for any $x\in\mX_f$, there exists $u\in\mU$ such that $f(x,u)\in\mX_f$.
    \item For any $x_0\in\mX_0$, there exists at least one sequence $u_t\in\mU$ such that $x_t\in\mC\ominus\mathcal{B}^n\left(r_{\delta,t},0\right)$ hold for any $t=0,\dots,T$, where $x_{t+1}=f(x_t,u_t)$.
\end{enumerate}
\end{assumption}
With this assumption, the following theorem derives the recursive feasibility and $1-\delta$ safety guarantee of the proposed MPC scheme. 

\smallskip
\begin{thm} \label{thm 1} \cite[Thm. 2]{kohler2024predictive}
    Given stochastic system \eqref{sys: d-t ss}, safe set $\mC$, cost functions $\mL_t(\cdot,\cdot)$, terminal cost $\Phi(\cdot)$, initial state set $\mX_0\in\mC$, safe probability level $1-\delta$, and terminal state set $\mX_f$, under Assumptions \ref{ass: Lipschitz f}-\ref{ass: Xf}, the following statements hold: 
    \begin{itemize}
        \item [a)] (Recursive feasibility) Suppose that  problem \eqref{eq: mpc scheme} has a feasible solution $\bm{u}_t^*$ at the time $t$, then there exists a solution $\bm{u}_{t+1}$ of \eqref{eq: mpc scheme} at time $t+1$. 
        \item [b)] (Safety guarantee) Suppose that problem \eqref{eq: mpc scheme} has a feasible solution $\bm{u}^*_0$ at $t=0$, then by choosing $u_t=u_{t|t}^*$, it holds that $\prob{X_t\in\mC,~\forall t\leq T}\geq 1-\delta$.
    \end{itemize}
\end{thm}

\subsection{Recursive Feasibility and Safety Guarantee}
Recursive feasibility and safety guarantee are the foundation of an MPC scheme in the safe control problem. For completeness, by applying the classical sliding-window technique \cite{rawlings2020model}, we present a proof to show the recursive feasibility of the proposed MPC (Theorem \ref{thm 1}). Based on the recursive feasibility, we prove that the control input generated by the proposed MPC can keep the system \eqref{sys: d-t ss} safe with $1-\delta$ guarantee.

\subsubsection{Proof of Theorem \ref{thm 1}-a)}

    Suppose that problem \eqref{eq: mpc scheme} has a feasible solution $\bm{u}_t^*$ at the time slot $t$. It suffices to construct a feasible solution $\tilde{\bm{u}}_{t+1}$ that satisfies all the constraints of \eqref{eq: mpc scheme} at time $t+1$. Using the sliding-window technique, a candidate solution $\tilde{\bm{u}}_{t+1}$ can be constructed as $\tilde{\bm{u}}_{t+1}=\{u^*_{t+1|t},\dots, u^*_{t+m|t}, \tilde{u}_f\}$, where $\tilde{u}_f$ ensures that $f(x^*_{t+m|t},\tilde{u}_f)\in\mX_f$. Since $\tilde{x}_{t+1|t+1}=x_{t+1}$ always holds, and $z_{t+k|t+1}$ can be always calculated through (\ref{eq: mpc scheme}a-b) given any $\tilde{\bm{u}}_{t+1}$, it is sufficient to prove that the state sequence $\tilde{x}_{t+1+k|t+1}$, $k=1,\dots,m$ controlled by the constructed $\tilde{u}_{t+1}$ satisfies (\ref{eq: mpc scheme}f) and (\ref{eq: mpc scheme}g). Let $\tilde{\bm{x}}_{t}^*=\{\tilde{x}_{t|t}^*,\dots, \tilde{x}_{t+m|t}^*\}$ be the solution of $\tilde{x}$ at $t$, then we know $\tilde{x}_{t+1+k|t+1}=\tilde{x}^*_{t+1+k|t}$, $\forall k=1,\dots,m-1$, and thus they satisfy (\ref{eq: mpc scheme}f). As for the terminal state $\tilde{x}_{t+k+m|t+1}$, it belongs to $\mX_f$ by Assumption \ref{ass: Xf} on $\mX_f$. This completes the proof.   

\subsubsection{Proof of Theorem \ref{thm 1}-b)}
Suppose that problem \eqref{eq: mpc scheme} has feasible solutions $\bm{u}^*_t$ at each time $t=0,\dots,T-1$. Since $u_t=u_{t|t}^*$ is applied to $X_t$ and $x_t$ simultaneously and $x_t$ is noise-free, by (\ref{eq: mpc scheme}f), it holds that 
\begin{equation} \label{eq: safe proof 1}
    x_t\in\mC\ominus \mathcal{B}^n(r_{\delta,t},0), ~\forall t\leq T.
\end{equation}
By the definition of $r_{\delta,t}$ given in \eqref{eq: def PT}, we obtain 
\begin{equation}  \label{eq: safe proof 2}
    \prob{ X_t\in \{x_t\}\oplus\mathcal{B}^n\left(r_{\delta,t},0\right), \forall t\leq T}\geq 1-\delta.
\end{equation}
Combining \eqref{eq: safe proof 1} and \eqref{eq: safe proof 2}, we conclude that $\prob{X_t\in\mC,~\forall t\leq T}\geq 1-\delta$. This completes the proof.

It is clear that Assumption \ref{ass: Xf} is necessary for the proof of Theorem \ref{thm 1}. Essentially, the hold of Assumption \ref{thm 1} relies on a \textit{tight} radius $r_{\delta,t}$ of PT. That is, given all the parameters, the value of $r_{\delta,t}$ should be as small as we can achieve. If $r_{\delta,t}$ is overly conservative, then the eroded set $\mC\ominus\mathcal{B}^n\left(r_{\delta,t},0\right)$ can be restrictive or even empty, making the constraints (\ref{eq: mpc scheme}f-g) infeasible. Therefore, a sufficiently tight $r_{\delta,t}$ is crucial. 

\subsection{A Tight Radius of PT}
In \cite[Theorem 5]{liu2025safety}, a tight radius $r_{\delta,t}$ of the PT is 
\begin{equation} \label{eq: r=}
\begin{cases}
    \sigma\left(\sqrt{\tfrac{1-L^{2t}}{1-L^2}}+\sqrt{\tfrac{L^{-2(\Delta t-1)}-1}{L^{-2}-1}}\right)\sqrt{\varepsilon_1n+\varepsilon_2\log\frac{2T}{\delta \Delta t}},~ L<1, \\
    L^t\sigma\sqrt{\frac{L^{-2T}-1}{L^{-2}-1}(\varepsilon_1n+\varepsilon_2\log(1/\delta))},~ L\geq 1,
\end{cases}
\end{equation}
where $\varepsilon_1=\frac{-\log(1-\varepsilon^2)}{\varepsilon^2},~ \varepsilon_2=\frac{2}{\varepsilon^2}$ for any user-chosen $\varepsilon\in(0,1)$, and $\Delta t\in\mathbb{N}_+$ can be either chosen as a small positive integer or optimized given other parameters.  
Note that $r_{\delta,t}$ in \eqref{eq: r=} is imposed on the entire stochastic trajectory, so its use yields the $1-\delta$ safety guarantee at the trajectory level, rather than at a single time point (e.g., the bound used in \cite{kohler2024predictive}). Moreover, the derived $r_{\delta,t}$ only has an $\mO(\sqrt{\log(1/\delta)})$ dependence on $\delta$, making it particularly superior when $\delta$ is small (e.g., $\delta=10^{-6}$).
Compared to existing works, the $\mO(\sqrt{\log(1/\delta)})$ dependence achieved by $r_{\delta,t}$ represents the tightest known result. 
Additionally, when $L<1$, the term with respect to $L$ is bounded, and $r_{\delta,t}$ only has an $\tmO(\sqrt{\log T})$ dependence on the terminal time $T$, making it effective even if $T$ is large. Therefore, we usually need $L<1$ in practice.

\subsection{Justification on Lipschitz Constants}

One may notice that the adopted bound \eqref{eq: r=} is related to the open-loop Lipschitz constant $L$ given in Assumption \ref{ass: Lipschitz f}. 
However, since $u_t$ generated by the presented MPC is a feedback of $X_t$, the Lipschitz constant of the closed-loop system $X_{t+1}=f(X_t,u(X_t))+w_t$ is no longer equal to $L$ and is difficult to capture. Arising from this difference, whether to use the open-loop or closed-loop Lipschitz constant in $r_{\delta,t}$ remains unclear in existing works \cite{kohler2024predictive}.
In this section, we justify that the bound $r_{\delta,t}$ derived in \eqref{eq: r=} only uses the Lipschitz constant of \textit{the original open-loop system}, and the change of the closed-loop Lipschitz constant is irrelevant.  

Recall that the actual stochastic trajectory $\{X_t\}$ and the nominal trajectory $\{x_t\}$ has the same control input $u_t$ at any time. By Assumption \ref{ass: Lipschitz f}, we know that given a time $t$:
\begin{equation}\label{eq: L X-x}
    \|f(X_t,u_t)-f(x_t,u_t)\|\leq L\|X_t-x_t\|.
\end{equation}
Note that \eqref{eq: L X-x} holds for any $u_t\in\mU$, even if $u_t$ is a feedback control $u(X_t)$ of $X_t$, because the same $u_t$ is applied to the associated $X_t$ and $x_t$ simultaneously. This implies that the evolution of $X_t-x_t$ is \textit{determined} by the open-loop Lipschitz constant $L$, and the Lipschitz constant of the closed-loop system is irrelevant. Then, since $r_{\delta,t}$ defined as \eqref{eq: def PT} is a probabilistic bound on the gap $\|X_t-x_t\|$, its temporal curve is related to the evolution of $X_t-x_t$. Therefore, the expression of $r_{\delta,t}$ contains $L$, and the closed-loop Lipschitz constant, which is irrelevant to the evolution of $X_t-x_t$, does not appear. A detailed derivation of $r_{\delta,t}$ shown in \eqref{eq: r=} can be found in \cite[Section VI]{liu2025safety}.

This justification is crucial for safe MPC design, as it assures that one can directly use the open-loop Lipschitz $L$ constant when implementing the set erosion strategy, without worrying about the implicit effect of the feedback. When the open-loop system \eqref{sys: d-t ss} is given, $L$ can be estimated and tuned by various existing techniques as stated in Remark \ref{remark}.
A linear system is displayed below to further exemplify our conclusion. 

\begin{example}
    Consider the associated linear trajectories:
\begin{equation}
    \begin{split}
        &X_{t+1}=AX_t+Bu(X_t)+w_t \\
        &x_{t+1}=Ax_t+Bu(X_t),
    \end{split}
\end{equation}
where $w_t$ is stochastic and $u(X_t)$ is a nonlinear feedback of $X_t$. The closed-loop dynamics $AX_t+Bu(X_t)$ is nonlinear. Interestingly, it holds that
\begin{equation}
    X_{t+1}-x_{t+1}=A(X_t-x_t)+w_t.
\end{equation}
Let $e_t=X_t-x_t$, then $e_t$ is still driven by the original linear system $e_{t+1}=Ae_t+w_t$. Therefore, for any control law $u(X_t)$, the probabilistic tube around $x_t$ is only determined by $A$ and $w_t$. 
\end{example}

\subsection{Performance Analysis for Lipschitz Cost}
The proposed MPC scheme \eqref{eq: mpc scheme} imposes safety constraints on the nominal trajectory $\{x_t\}$, while $u_t$ is implemented on the actual trajectory $\{X_t\}$ under stochastic disturbance. To understand how the performance of $\{X_t\}$ is related to $\{x_t\}$, we quantify the gap between $J(\bm{X,u})=\mbE\left(\sum_{t=0}^{T-1}\mL_t(X_t,u_t)+\Phi(X_T)\right)$ and $\tilde{J}(\bm{x,u})=\mbE(\sum_{t=0}^{T-1}\mL_t(x_t,u_t)$$+\Phi(x_T))$.
In this section, the cost functions are assumed to be Lipschitz. More general cases are considered as future work.

\begin{thm} \label{thm 2}
Given stochastic system \eqref{sys: d-t ss}, safe set $\mC$, cost functions $\mL_t(\cdot,\cdot)$, terminal cost function $\Phi(\cdot)$, initial state set $\mX_0\in\mC$, safe probability level $1-\delta$, and terminal state set $\mX_f$. Suppose that the cost functions are $L_c$-Lipschitz with both $x$ and $u$, and Assumption \ref{ass: Lipschitz f}-\ref{ass: Xf} hold. Let $u_t$ be the control input of the actual trajectory $\{X_t\}$. Define $J(\bm{X,u})$ and $\tilde{J}(\bm{x,u})$ as above in Section III-E. Then it holds that: 
\begin{equation} \label{eq: thm 2}
\begin{split}
    J(\bm{X,u})-\tilde{J}(\bm{x,u}) 
    \leq \sum_{t=0}^T\sqrt{\frac{nL_c^2\sigma^2(L^{2t}-1)}{L^2-1}}.
\end{split}
\end{equation}
\end{thm}

\begin{proof}
As derived in \cite[Proposition 2]{szy2024Auto}, the evolution of $\mbE(\|X_t-x_t\|^2)$ satisfies:
\begin{equation} \label{eq: EV_t+1<=}
\begin{split}
     \mbE\left(\|X_{t+1}-x_{t+1}\|^2\right)\leq L^2\mbE\left(\|X_{t}-x_{t}\|^2\right)+n\sigma^2
\end{split}
\end{equation}
which has a solution 
\begin{equation} \label{eq: E-bound}
    \mbE(\|X_t-x_t\|^2)\leq \frac{n\sigma^2(L^{2t}-1)}{L^2-1}. 
\end{equation}
 Since $\mL_t(x,u)$ and $\Phi(x)$ are $L_c$-Lipschitz w.r.t. $x$, we know $\mbE(\mL_t(X_t,u_t)-\mL_t(x_t,u_t))\leq L_c\mbE(\|X_t-x_t\|)$,
and the same inequality holds for $\Phi(\cdot)$. By \eqref{eq: E-bound}, we arrive the conclusion of \eqref{eq: thm 2}.
This completes the proof.
\end{proof}
When $L>1$, the term $L^T$ dominates the gap $J(\bm{X,u})-\tilde{J}(\bm{x,u})$, which grows exponentially with respect to $T$. When $L=1$, this gap reduces to the level of $\mO(T^{\frac{3}{2}})$. When $L<1$, the term $\sum_{t=0}^T\sqrt{\frac{nL_c^2\sigma^2(L^{2t}-1)}{L^2-1}}$ converge to some constant, so $J(\bm{X,u})-J^*(\bm{x}^*,\bm{u}^*)$ is on the level of $\mO(T)$. If we take into account the average cost gap $\frac{1}{T}\left(J(\bm{X,u})-J^*(\bm{x}^*,\bm{u}^*)\right)$, then it diverges with $T$ when $L\geq1$ and converges to a finite value when $L<1$. This implies that when $L<1$, the performance of $\{X_t\}$ and $\{x_t\}$ only has a bounded gap under stochastic noise.

\section{Numerical Examples}\label{sec: case}
This section presents two examples to illustrate the proposed stochastic MPC method. The solver we choose for both examples is the Optimization Toolbox in MATLAB.
\subsection{Unicycle}

Consider a unicycle moving on a $2$-dimensional plane with obstacles. The unsafe region $\mC_u$ is the union of red obstacles shown in Figure~\ref{fig: unicycle}, and the safe region is $\mC=\R^n\setminus\mC_u$. Define $d_{min}$ as the minimal width of the path between two circular obstacles. 
The discrete-time system model is given as
\begin{equation}\label{sys: uni}
    \begin{split}
        X_{t+1} = X_t + \eta\begin{bmatrix}
    (v_t^{st}(X_t)+u_{1,t}) \cos(\theta_t)\\
    (v_t^{st}(X_t)+u_{1,t}) \sin(\theta_t)\\
    \omega_t^{st}(X_t)+u_{2,t}
\end{bmatrix} + w_t,
    \end{split}
\end{equation}
where $X_t = \begin{bmatrix}p_{x,t} & p_{y,t} & \theta_t\end{bmatrix}^{\top}$ is the state of the vehicle, $(p_{x,t},p_{y,t})$ is the position of the center of mass of the vehicle in the plane, $\theta_t$ is the heading angle of the vehicle, $\eta=0.1$ is the discretization step size, and  $w_t\sim\mathcal{N}(0,0.02\eta)$ is the stochastic disturbance on the model. The velocity $v_t$ and the angular velocity $\omega_t$ of the unicycle are designed as $v_t = v_t^{st}(X_t)+u_{1,t}$ and $\omega_t = \omega_t^{st}(X_t)+u_{2,t}$,
where $v_t^{st}(X_t)$ and $\omega_t^{st}(X_t)$ are functions of $X_t$ to stabilize the system, and $u_t=\begin{bmatrix} u_{1,t} & u_{2,t} \end{bmatrix}$ is the control input that is generated by the proposed MPC scheme. The expressions of $v_t^{st}(X_t)$ and $\omega_t^{st}(X_t)$ are chosen following \cite{MA-GC-AB-AB:95}, with which the Lipschitz constant of $f$ w.r.t. $X_t$ is estimated to be 0.96 by using the sampling-based method \cite{chuchu2017simulation} over the displayed safe region. The task is to control the unicycle from $x_0=[2.2~3.6~ \frac{\pi}{3}]^\top$ to the origin during $t=0,\dots,30$ with 99.9\% safety guarantee ($\delta=10^{-3}$). The presented MPC scheme is applied with $m=20$, $\mU=\{u|\|u\|_\infty\leq2\}$ and loss functions $\mL_t(x,u)=a\|x\|_1+b\|u\|_1$, $\Phi(x)=a\|x\|$ with constants $a,b$. We sampled 1000 independent stochastic trajectories, and compare their cost $\mL_t(X_t,u_t)$ with $\mbE(\mL_t(X_t,u_t))$ as well as the deterministic cost in the noise-free environment. It is clear that all the sampled trajectories are safe, and the actual costs are close to both the nominal and deterministic costs. Moreover, Figure \ref{fig: unicycle}(c) indicates that without any safety constraints, the actual trajectories can be dangerous when the nominal trajectory is safe but violates the eroded area of the safe set. Finally, Figure \ref{fig: unicycle}(d) shows that our MPC scheme with $r_{\delta,t}$ in \eqref{eq: r=} is feasible with a reasonable safe path width $d_{min}$ when $\delta<10^{-4}$, while that determined by \cite{kohler2024predictive} can be exponentially large. In this simulation environment where $d_{min}=1.04$, our MPC scheme with $r_{\delta,t}$ in \eqref{eq: r=} is feasible for $\delta=10^{-5}$, while that determined by \cite{kohler2024predictive} can lead to infeasible solutions when $\delta\leq0.02$.

\begin{figure}[t] 
	\centering
  \includegraphics[width =0.46\linewidth]{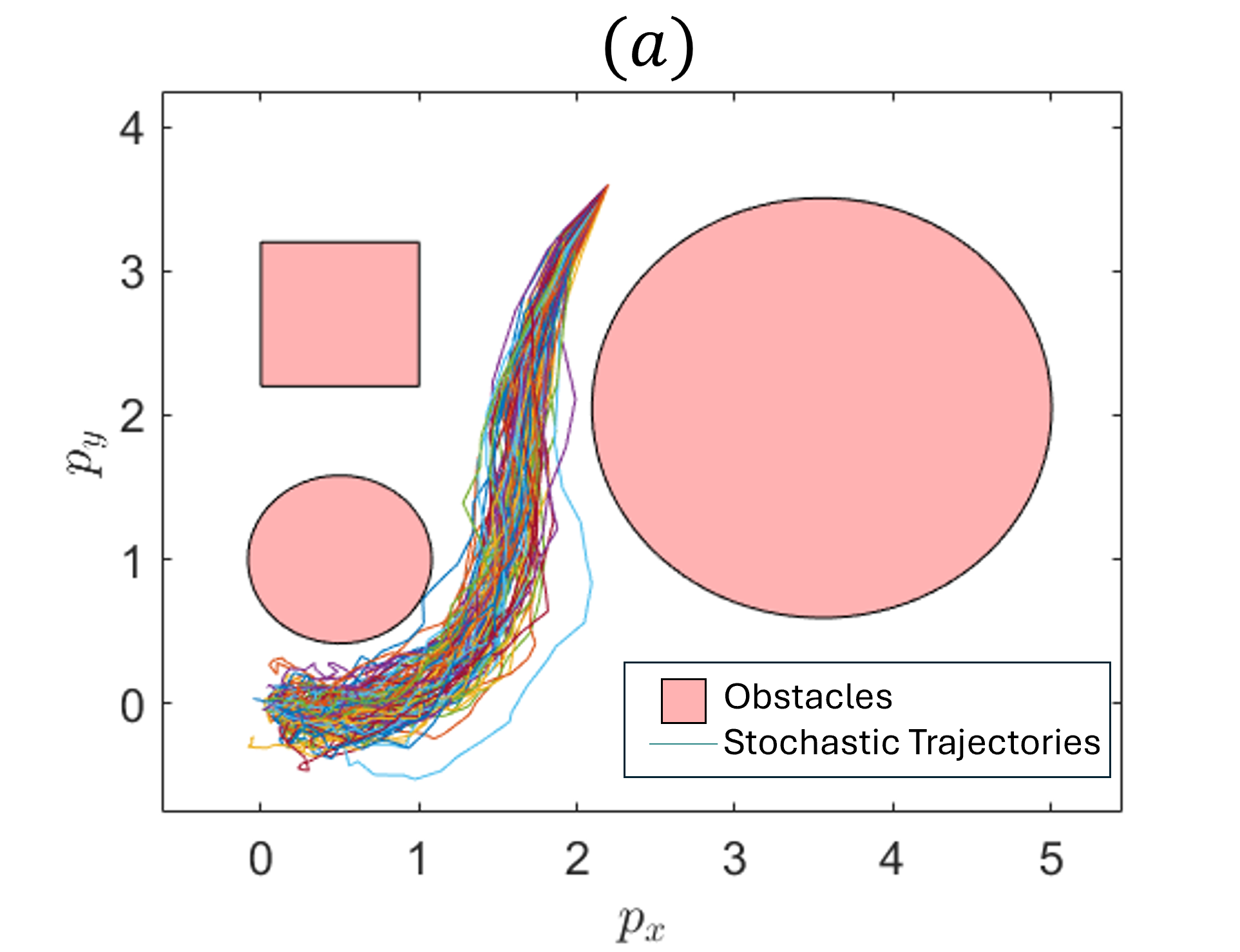}
  \includegraphics[width =0.46\linewidth]{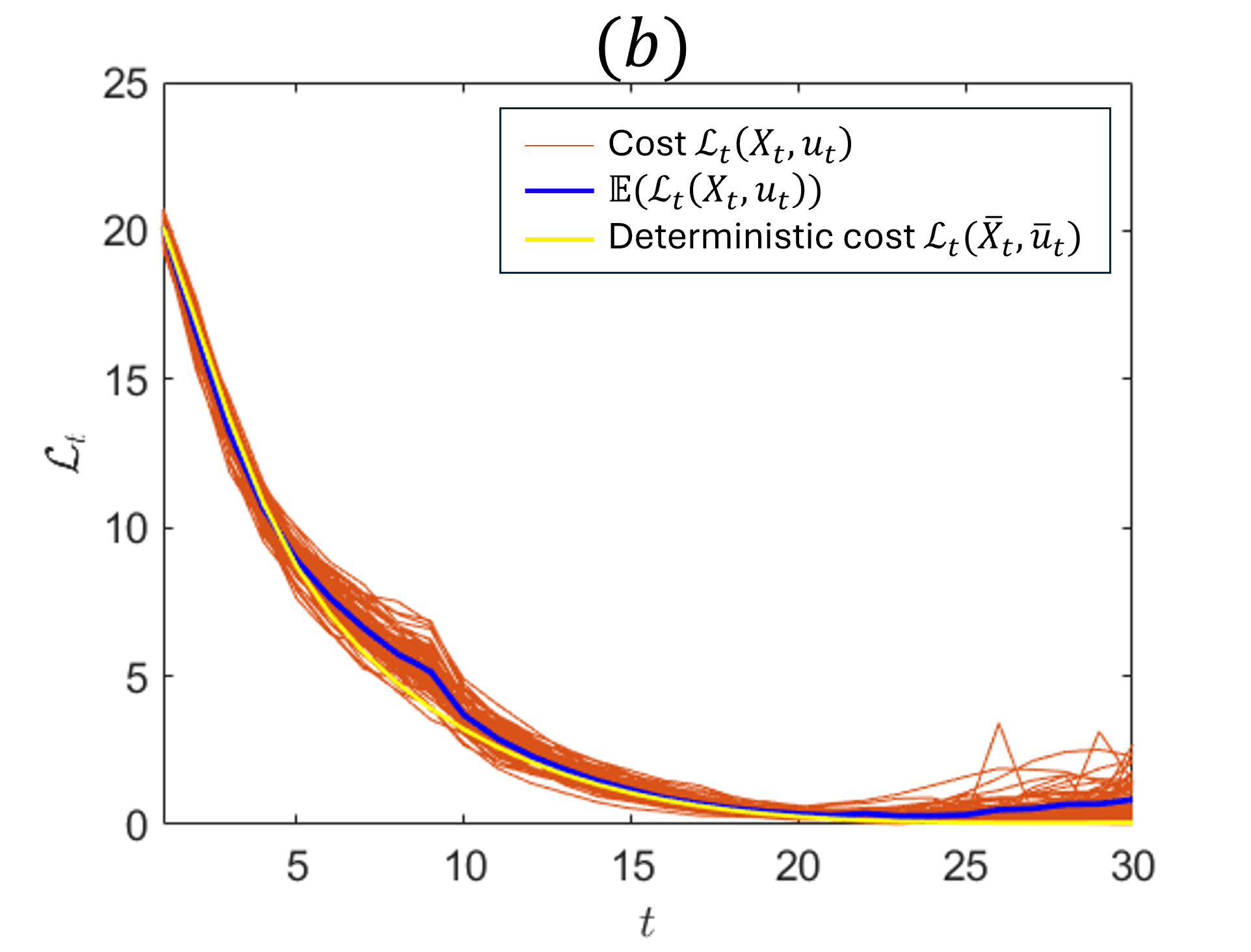}
  \includegraphics[width =0.49\linewidth]{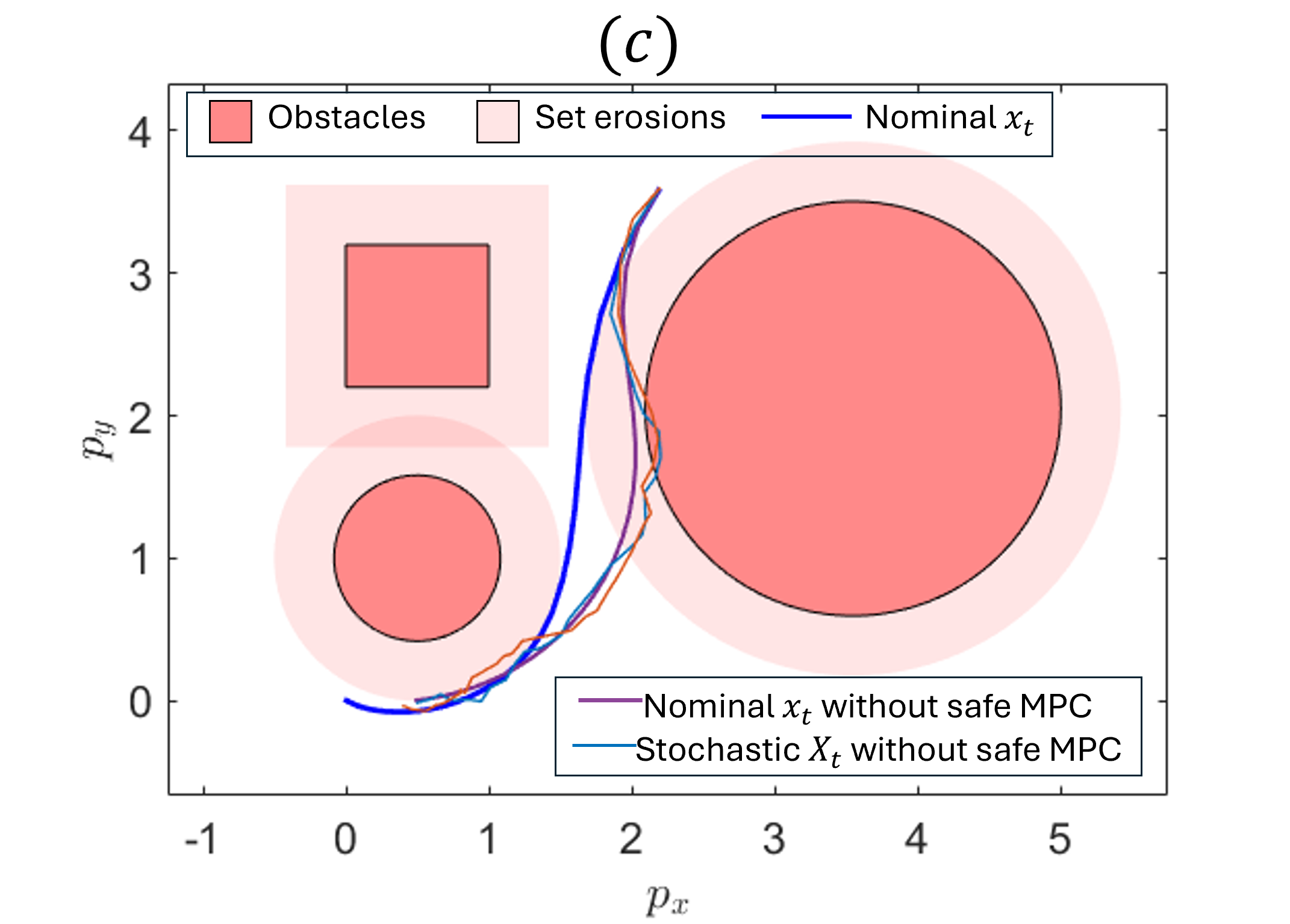}
  \includegraphics[width =0.46\linewidth]{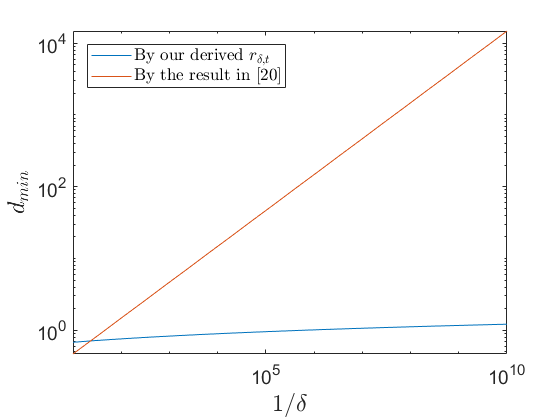}
	\caption{Experiments on a unicycle modeled by \eqref{sys: uni}. (a): 1000 sampled stochastic trajectories of \eqref{sys: uni} under the control of the proposed MPC. (b): Instantaneous cost $\mL_t(X_t,u_t)$ (in red) of the 1000 trajectories in (a), compared with $\mbE(\mL_t(X_t,u_t))$ (in blue) and $\mL_t(\bar{x}_t,\bar{u}_t)$ defined in Section III.E. (c) The nominal trajectory $\{x_t\}$ controlled by the proposed MPC (in blue), which is compared with the eroded safe set (outside the pink areas), the nominal trajectory $\{x_t^{st}\}$ with stabilizers only (in purple), and the associated stochastic trajectories of $x_t^{st}$. 
 }
 \label{fig: unicycle}
 \end{figure}

\subsection{2D Quadrotor}
\begin{figure}[t] 
	\centering
  \includegraphics[width =0.45\linewidth]{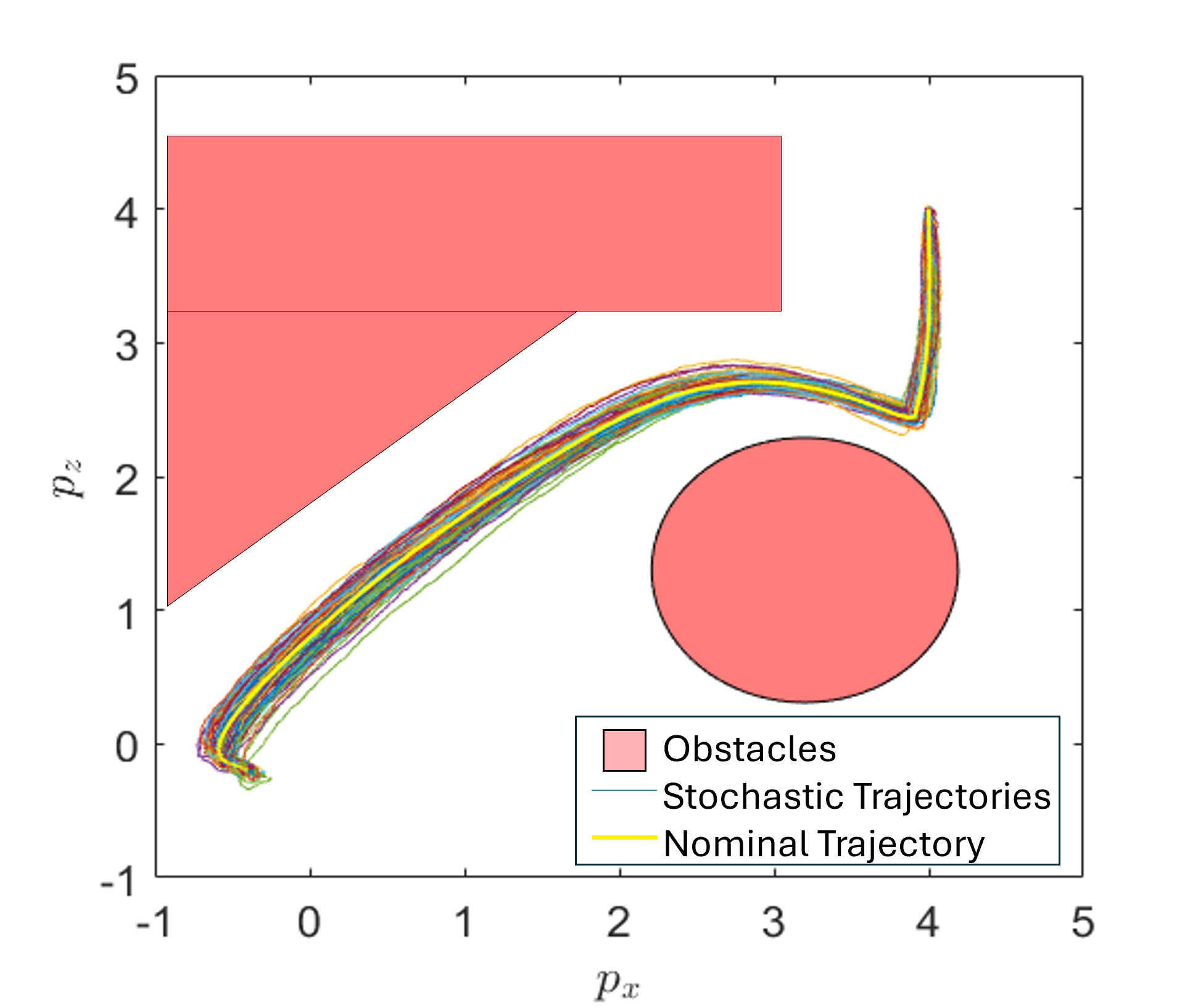}
  \includegraphics[width =0.5\linewidth]{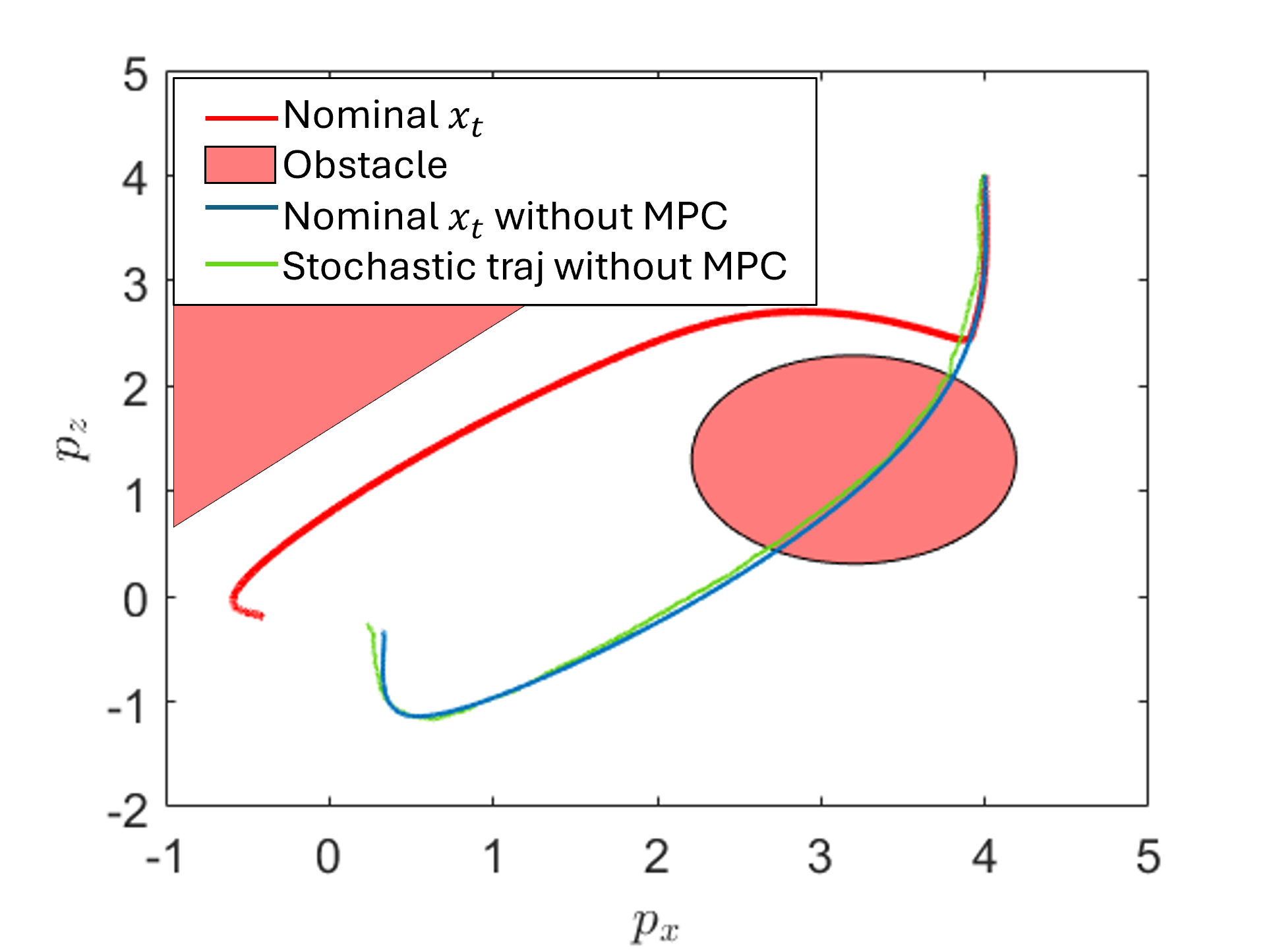}
	\caption{Experiments on the quadrotor. Left: 1000 independent stochastic trajectories of the quadrotor. 
    Right: Comparison between the nominal trajectory $\{x_t\}$ (in red) and that without safety MPC control input (in blue, its associated stochastic trajectory is in green).
 }
 \label{fig: quadrotor}
 \end{figure}

Consider a 6-dimensional planar quadrotor. Let $(p_x,p_z)$ be the position of the quadrotor's center of mass in the plane, $\theta$ be the orientation (tilt angle) of the quadrotor relative to the vertical, $(v_x,v_z)$ be the velocity of the quadrotor in the plane, $\dot{\theta}$ be the angular velocity. 
 Let $X_t=\begin{bmatrix} p_x&p_z&\theta&v_x&v_z&\dot{\theta} \end{bmatrix}^{\top}$ be the state vector, then the state space model is formulated as 
\begin{equation*}\label{sys: PQ}
    \begin{split}
        X_{t+1}=X_t+\eta\begin{bmatrix}
            v_x\cos\theta-v_z\sin\theta \\
            v_x\sin\theta+v_z\cos\theta \\
            \dot{\theta} \\
            v_z\dot{\theta}-g\sin\theta \\
            -v_x\dot{\theta}-g\cos\theta+\frac{K_{1}(X_t)X_t+u_{1,t}}{m_q} \\
            \frac{l}{J}(K_{2}(X_t)X_t+u_{2,t})
        \end{bmatrix}+w_t,
    \end{split}
\end{equation*}
$K_1(X_t),~K_2(X_t)$ are the stabilizers calculated by $X_t$, $u_t=[u_{1,t} ~u_{2,t}]^\top$ is the control input on the propellers, $w_t$ is the stochastic disturbance acting on the system, and $\eta=0.001$ is the discretization step size. We suppose that $g=9.8$ is the gravity, $l=0.25$ is the arm length, $J=0.035$ is the moment of inertia, $m_q=0.141$ is the mass. The bounded input set $\mU=\{u\in\R^2~|~\|u\|_\infty\leq 3\}$, and $w_t\sim\mathcal{N}(0,0.05\eta I_6)$. 

We conduct experiments on the time interval $t=0,\dots,2000$. The task is to use our presented MPC to control the quadrotor trajectory to the origin while avoiding the obstacles with probability at least $99.9\%$ ($\delta=10^{-3}$). The stabilizer $[K_1(X_t),~K_2(X_t)]$ is designed based on GVI-MP proposed in \cite{yu2023stochastic}, which make $L\approx 0.9985$. We set $m=20$ and the cost functions as the quadratic loss. We sampled 1000 independent stochastic trajectories.
All the sampled trajectories are safe. 
Also, the right side of Figure \ref{fig: quadrotor} indicates that with only the stabilizers, the actual trajectories can violate the safety constraint.

\section{Conclusion} \label{sec: conclusion}
We present an MPC framework for stochastic nonlinear systems with a high-probability safety guarantee. The framework is based on the set erosion strategy, which converts the chance constraints on the safe set to deterministic safety constraints of the nominal trajectory on an eroded subset of the safe set. When utilizing this strategy, we adopt a tight bound on the erosion depth of the safe set. Compared to existing works like \cite{kohler2024predictive}, it provides provides safety guarantee for the entire trajectory and significantly improves the feasibility of the proposed MPC when the level of safety guarantee is high. Moreover, we justify that the bound on the set erosion depth depends on the open-loop Lipschitz constant, rather than the closed-loop Lipschitz constant.

\bibliographystyle{ieeetr}
\bibliography{main}    

\begin{thebibliography}{10}

\bibitem{lavaei2022automated}
A.~Lavaei, S.~Soudjani, A.~Abate, and M.~Zamani, ``Automated verification and
  synthesis of stochastic hybrid systems: A survey,'' {\em Automatica},
  vol.~146, p.~110617, 2022.

\bibitem{kouvaritakis2015model}
B.~Kouvaritakis and M.~Cannon, {\em Model predictive control}.
\newblock Switzerland: Springer International Publishing, 2015.

\bibitem{ames2019control}
A.~D. Ames, S.~Coogan, M.~Egerstedt, G.~Notomista, K.~Sreenath, and P.~Tabuada,
  ``Control barrier functions: Theory and applications,'' in {\em 2019 18th
  European control conference (ECC)}, pp.~3420--3431, IEEE, 2019.

\bibitem{bajcsy2019efficient}
A.~Bajcsy, S.~Bansal, E.~Bronstein, V.~Tolani, and C.~J. Tomlin, ``An efficient
  reachability-based framework for provably safe autonomous navigation in
  unknown environments,'' in {\em 2019 IEEE 58th Conference on Decision and
  Control (CDC)}, pp.~1758--1765, IEEE, 2019.

\bibitem{morari1999model}
M.~Morari and J.~H. Lee, ``Model predictive control: past, present and
  future,'' {\em Computers \& chemical engineering}, vol.~23, no.~4-5,
  pp.~667--682, 1999.

\bibitem{lew2024risk}
T.~Lew, M.~Greiff, F.~Djeumou, M.~Suminaka, M.~Thompson, and J.~Subosits,
  ``Risk-averse model predictive control for racing in adverse conditions,''
  {\em arXiv preprint arXiv:2410.17183}, 2024.

\bibitem{knoedler2025safety}
L.~Knoedler, O.~So, J.~Yin, M.~Black, Z.~Serlin, P.~Tsiotras, J.~Alonso-Mora,
  and C.~Fan, ``Safety on the fly: Constructing robust safety filters via
  policy control barrier functions at runtime,'' {\em IEEE Robotics and
  Automation Letters}, 2025.

\bibitem{santoyo2021barrier}
C.~Santoyo, M.~Dutreix, and S.~Coogan, ``A barrier function approach to
  finite-time stochastic system verification and control,'' {\em Automatica},
  vol.~125, p.~109439, 2021.

\bibitem{fushimi2025safety}
S.~Fushimi, K.~Hoshino, and Y.~Nishimura, ``Safety-critical control for
  discrete-time stochastic systems with flexible safe bounds using affine and
  quadratic control barrier functions,'' {\em arXiv preprint arXiv:2501.09324},
  2025.

\bibitem{APV-MKO:21}
A.~P. Vinod and M.~K. Oishi, ``Stochastic reachability of a target tube: Theory
  and computation,'' {\em Automatica}, vol.~125, p.~109458, 2021.

\bibitem{liniger2017racing}
A.~Liniger, X.~Zhang, P.~Aeschbach, A.~Georghiou, and J.~Lygeros, ``Racing
  miniature cars: Enhancing performance using stochastic mpc and disturbance
  feedback,'' in {\em 2017 American Control Conference (ACC)}, pp.~5642--5647,
  IEEE, 2017.

\bibitem{yin2025safe}
J.~Yin, O.~So, E.~Y. Yu, C.~Fan, and P.~Tsiotras, ``Safe beyond the horizon:
  Efficient sampling-based mpc with neural control barrier functions,'' {\em
  Robotics: Science and Systems (RSS)}, 2025.

\bibitem{6426462}
M.~Prandini, S.~Garatti, and J.~Lygeros, ``A randomized approach to stochastic
  model predictive control,'' in {\em 2012 IEEE 51st IEEE Conference on
  Decision and Control (CDC)}, pp.~7315--7320, 2012.

\bibitem{wu2022robust}
A.~Wu, T.~Lew, K.~Solovey, E.~Schmerling, and M.~Pavone, ``Robust-rrt:
  Probabilistically-complete motion planning for uncertain nonlinear systems,''
  in {\em The International Symposium of Robotics Research}, pp.~538--554,
  Springer, 2022.

\bibitem{yu2023stochastic}
H.~Yu and Y.~Chen, ``Stochastic motion planning as gaussian variational
  inference: Theory and algorithms,'' {\em arXiv preprint arXiv:2308.14985},
  2023.

\bibitem{farina2016stochastic}
M.~Farina, L.~Giulioni, and R.~Scattolini, ``Stochastic linear model predictive
  control with chance constraints--a review,'' {\em Journal of Process
  Control}, vol.~44, pp.~53--67, 2016.

\bibitem{ao2025stochastic}
Y.~Ao, J.~K{\"o}hler, M.~Prajapat, Y.~As, M.~Zeilinger, P.~F{\"u}rnstahl, and
  A.~Krause, ``Stochastic model predictive control for sub-gaussian noise,''
  {\em arXiv preprint arXiv:2503.08795}, 2025.

\bibitem{liu2024safety}
Z.~Liu, S.~Jafarpour, and Y.~Chen, ``Safety verification of stochastic systems:
  A set-erosion approach,'' {\em IEEE Control Systems Letters}, 2024.

\bibitem{hewing2020recursively}
L.~Hewing, K.~P. Wabersich, and M.~N. Zeilinger, ``Recursively feasible
  stochastic model predictive control using indirect feedback,'' {\em
  Automatica}, vol.~119, p.~109095, 2020.

\bibitem{kohler2024predictive}
J.~K{\"o}hler and M.~N. Zeilinger, ``Predictive control for nonlinear
  stochastic systems: Closed-loop guarantees with unbounded noise,'' {\em IEEE
  Transactions on Automatic Control}, 2025.

\bibitem{manchester2017control}
I.~R. Manchester and J.-J.~E. Slotine, ``Control contraction metrics: Convex
  and intrinsic criteria for nonlinear feedback design,'' {\em IEEE
  Transactions on Automatic Control}, vol.~62, no.~6, 2017.

\bibitem{9683354}
G.~Chou, N.~Ozay, and D.~Berenson, ``Model error propagation via learned
  contraction metrics for safe feedback motion planning of unknown systems,''
  in {\em 2021 60th IEEE Conference on Decision and Control (CDC)},
  pp.~3576--3583, 2021.

\bibitem{vershynin2018high}
R.~Vershynin, {\em High-Dimensional Probability: An Introduction with
  Applications in Data Science}.
\newblock Cambridge Series in Statistical and Probabilistic Mathematics,
  Cambridge University Press, 2018.

\bibitem{7740982}
A.~Mesbah, ``Stochastic model predictive control: An overview and perspectives
  for future research,'' {\em IEEE Control Systems Magazine}, vol.~36, no.~6,
  pp.~30--44, 2016.

\bibitem{liu2025safety}
Z.~Liu, S.~Jafarpour, and Y.~Chen, ``Safety verification of nonlinear
  stochastic systems via probabilistic tube,'' {\em arXiv preprint
  arXiv:2503.03328}, 2025.

\bibitem{rawlings2020model}
J.~B. Rawlings, D.~Q. Mayne, M.~Diehl, {\em et~al.}, {\em Model predictive
  control: theory, computation, and design}, vol.~2.
\newblock Nob Hill Publishing Madison, WI, 2020.

\bibitem{szy2024Auto}
Z.~Liu, S.~Jafarpour, and Y.~Chen, ``Probabilistic reachability of
  discrete-time nonlinear stochastic systems,'' {\em arXiv preprint
  arXiv:2409.09334}, 2024.

\bibitem{MA-GC-AB-AB:95}
M.~Aicardi, G.~Casalino, A.~Bicchi, and A.~Balestrino, ``Closed loop steering
  of unicycle like vehicles via {Lyapunov} techniques,'' {\em IEEE Robotics \&
  Automation Magazine}, vol.~2, no.~1, pp.~27--35, 1995.

\bibitem{chuchu2017simulation}
C.~Fan, J.~Kapinski, X.~Jin, and S.~Mitra, ``Simulation-driven reachability
  using matrix measures,'' {\em ACM Trans. Embed. Comput. Syst.}, vol.~17, dec
  2017.

\end{thebibliography}
\end{document}